\pdfoutput=1
\documentclass[10pt, conference, letterpaper]{IEEEtran}
\IEEEoverridecommandlockouts
\usepackage{amsmath,graphicx,amssymb,amsthm}
\usepackage{cite}

\usepackage{dsfont}
\usepackage{url}
\usepackage{color}
\usepackage{epstopdf}
\usepackage{verbatim}
\usepackage{cuted}
\usepackage{lipsum,graphicx,subcaption}
\usepackage[utf8]{inputenc}
\usepackage[english]{babel}
\usepackage{caption}
\usepackage{enumitem}
\usepackage[ruled,vlined]{algorithm2e}
\DeclareCaptionLabelFormat{lc}{\MakeLowercase{#1}~#2}
\usepackage{balance}
\usepackage{float}
\usepackage{mathtools}

\DeclarePairedDelimiter\floor{\lfloor}{\rfloor}

\newtheorem{assumption}{Assumption}
\newtheorem{theorem}{Theorem}

\newtheorem{definition}{Definition}
\newtheorem{proposition}{Proposition}
\newtheorem{corollary}{Corollary}

\DeclareMathOperator*{\argmax}{arg\,max}
\DeclareMathOperator*{\argmin}{arg\,min}
\allowdisplaybreaks
\renewcommand{\figurename}{Fig.}
\usepackage[protrusion=true,expansion=true]{microtype}
\usepackage[font=small]{caption}
\let\emptyset\varnothing

\usepackage{tabularx,booktabs}
\newcolumntype{Y}{>{\centering\arraybackslash}X}
\usepackage{multirow}

\def\BibTeX{{\rm B\kern-.05em{\sc i\kern-.025em b}\kern-.08em
     T\kern-.1667em\lower.7ex\hbox{E}\kern-.125emX}}

\begin{document}
\title{Device Sampling for Heterogeneous Federated Learning: Theory, Algorithms, and Implementation\vspace{-0.1in}}

\author{\IEEEauthorblockN{\normalsize Su Wang\IEEEauthorrefmark{1}, Mengyuan Lee\IEEEauthorrefmark{2}, Seyyedali Hosseinalipour\IEEEauthorrefmark{1}, Roberto Morabito\IEEEauthorrefmark{3}, Mung Chiang\IEEEauthorrefmark{1}, and Christopher G. Brinton\IEEEauthorrefmark{1}}
\IEEEauthorblockA{\small \IEEEauthorrefmark{1}School of Electrical and Computer Engineering, Purdue University\\
\IEEEauthorrefmark{2}College of Information Science and Electronic Engineering, Zhejiang University\\
\IEEEauthorrefmark{3}Department of Electrical Engineering, Princeton University, and Ericsson Research\\
Email: \IEEEauthorrefmark{1}\{wang2506, hosseina, chiang, cgb\}@purdue.edu,
\IEEEauthorrefmark{2}mengyuan\textunderscore lee@zju.edu.cn,
\IEEEauthorrefmark{3}roberto.morabito@princeton.edu}
\vspace{-0.4in}}

\maketitle
\begin{abstract}
The conventional federated learning (FedL) architecture distributes machine learning (ML) across worker devices by having them train local models that are periodically aggregated by a server. FedL ignores two important characteristics of contemporary wireless networks, however: (i) the network may contain heterogeneous communication/computation resources, while (ii) there may be significant overlaps in devices' local data distributions. In this work, we develop a novel optimization methodology that jointly accounts for these factors via intelligent device sampling complemented by device-to-device (D2D) offloading. Our optimization aims to select the best combination of sampled nodes and data offloading configuration to maximize FedL training accuracy subject to realistic constraints on the network topology and device capabilities. Theoretical analysis of the D2D offloading subproblem leads to new FedL convergence bounds and an efficient sequential convex optimizer. Using this result, we develop a sampling methodology based on graph convolutional networks (GCNs) which learns the relationship between network attributes, sampled nodes, and resulting offloading that maximizes FedL accuracy. Through evaluation on real-world datasets and network measurements from our IoT testbed, we find that our methodology while sampling less than 5\% of all devices outperforms conventional FedL substantially both in terms of trained model accuracy and required resource utilization.

\end{abstract}
\section{Introduction}

\noindent  
The proliferation of smartphones, unmanned aerial vehicles (UAVs), and other devices comprising the Internet of Things (IoT) is causing an exponential rise in data generation and large demands for machine learning (ML) at the edge~\cite{8373692}. 
For example, sensor and camera modules on self-driving cars produce up to 1.4 terabytes of data per hour~\cite{carData1} with the objective of training ML models for intelligent navigation. The traditional paradigm in ML of centralized training at a server is often not feasible in such environments since (i) transferring these large volumes of data from the devices to the cloud imposes long transfer delays and (ii) users are sometimes unwilling to share their data due to privacy concerns~\cite{7498684}. 

Federated learning (FedL) is a recently proposed distributed ML technique aiming to overcome these challenges~\cite{mcmahan2017communication,konevcny2016federated}.
Under FedL, devices train models on their local datasets, typically by means of gradient descent, and a server periodically aggregates the parameters of local models to form a global model. This global model is then transferred back to the devices for the next round of local updates, as depicted in Fig.~\ref{fig:simpleFL}.
In conventional FedL, each device processes its own collected data and operates independently within an aggregation period. This will become problematic in terms of upstream device communication and local device processing requirements, however, as its implementations scale to networks consisting of millions of heterogeneous wireless devices~\cite{niknam2020federated,hosseinalipour2020federated}.

At the same time, device-to-device (D2D) communications that are becoming part of 5G and IoT can enable local offloading of data processing from resource hungry to resource rich devices~\cite{9155510}. Additionally, we can expect that for particular applications, the datasets collected across devices will contain varying degrees of similarity, e.g.,
images gathered by UAVs conducting surveillance over the same area~\cite{9084352,kairouz2019advances}. Processing similar data distributions at multiple devices adds overhead to FedL and an opportunity for efficiency improvement.

\begin{figure}[t]
\includegraphics[width=.43\textwidth]{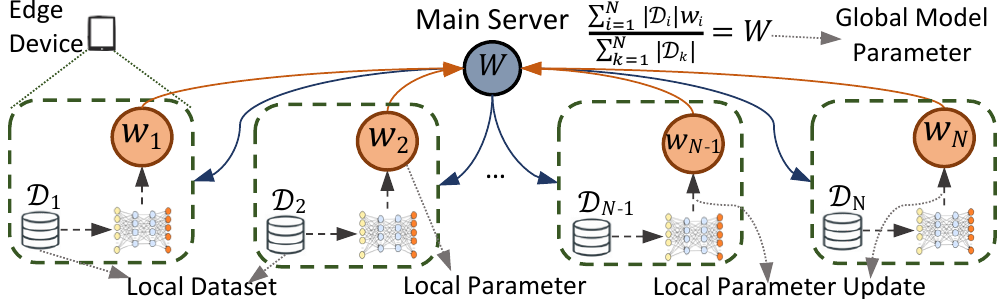}
\centering
\caption{Architecture of conventional federated learning (FedL).}
\label{fig:simpleFL}
\vspace{-6mm}
\end{figure}

Motivated by this, we develop a novel methodology for \textit{smart device sampling with data offloading} in FedL. Specifically, we formulate a joint sampling and data offloading optimization problem where devices expected to maximize contribution to model training are sampled for training participation, while devices that are not selected may transfer data to those that are. This data offloading is performed according to estimated data dissimilarities between nodes, which are updated as transfers are observed. We show that our methodology yields superior model performance to conventional FedL while significantly reducing network resource utilization. In our model motivated by paradigms such as \textit{fog learning}~\cite{hosseinalipour2020federated,hosseinalipour2020multi,tu2020network}, data offloading only occurs among trusted devices; devices that have privacy concerns are exempt from data offloading. 
\vspace{-1mm}
\subsection{Related Work}
To improve the communication efficiency of FedL, recent works have focused on 
efficient encoding designs to reduce parameter transmission sizes~\cite{9155479,sattler2019robust}, optimizing the frequency of global aggregations~\cite {8486403,wang2019adaptive}, and device sampling \cite{ji2020dynamic,nguyen2020fast}. Our work falls into the third category. In this regard, most works have assumed a static or uniform device selection strategy, e.g.,~\cite{wang2019adaptive,hosseinalipour2020multi,yang2019energy,tran2019federated,tu2020network,konevcny2016federated,sahu2018convergence,reisizadeh2020fedpaq,ji2020dynamic,nguyen2020fast}, 
where the main server chooses a subset of devices either uniformly at random or according to a pre-determined sampling distribution. There is also an emerging line of work on device sampling based on wireless channel characteristics, specifically in cellular networks~\cite{8851249,shi2019device,xia2020multi,ren2020scheduling,he2020importance}. By contrast, we develop a sampling technique that adapts to the heterogeneity of device resources and overlaps across local data distributions that are key characteristics of contemporary wireless edge networks. Also, we study device sampling based on the utility of device data distributions. Specifically, when compared to the limited literature on device sampling by each device's instantaneous contributions to global updates~\cite{9155494}, we introduce a novel perspective based on device data similarities. Our methodology exploits the proliferation of D2D communications at the wireless edge~\cite{tu2020network,hosseinalipour2020multi}, to diversify each selected device's local data via D2D offloading. Our work thus considers the novel problems of sampling and D2D offloading for FedL jointly, and leads to new analytical convergence bounds and algorithms used by implementations. 

It is worth mentioning two parallel lines of work in FedL that consider relationships between node data distributions. One is on fairness \cite{williamson2019fairness}, in which the objective is to train the ML model without biasing the result towards any one device's distribution, e.g.,~\cite{mohri2019agnostic,li2019fair}. Another leverages transfer learning techniques~\cite{pan2009survey} to build models across data parties (e.g., companies or enterprises) that possess partial overlaps~\cite{li2019fedmd,liu2018secure,gao2019privacy}. Our work is focused on a fundamentally different objective, i.e., network resource efficiency optimization.


\begin{figure}[t]
\includegraphics[width=.48\textwidth]{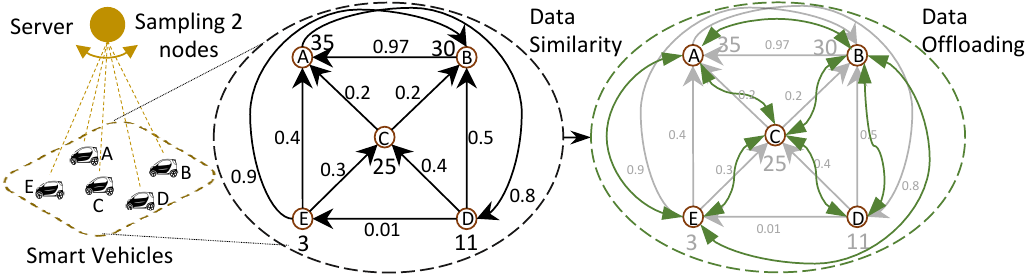}
\centering
\caption{Small motivating example of a wireless network composed of 5 smart cars and an edge server. The server can only sample two of the cars to participate in FedL training.}
\label{fig:SimilarOffload}
\vspace{-5mm}
\end{figure}
\vspace{-1mm}
\subsection{Motivating Toy Example}\label{sec:toy}
Consider Fig.~\ref{fig:SimilarOffload}, wherein five heterogeneous smart cars communicate with an edge server to train an object detection model. Due to limited bandwidth, the server can only exploit $2$ out of the $5$ cars to conduct FedL training, but needs to train a model representative of the entire dataset within this network. 
The computational capability of each car, i.e., the number of processed datapoints in one aggregation period, is shown next to itself, and the edge weights in the data similarity graph capture the similarity between the local data of the cars. 
Rather than using statistical distance metrics~\cite{liese2006divergences}, which are hard to compute in this distributed scenario, the data similarities could be estimated by commute routes and geographical proximity~\cite{jeske2013floating}. Also, in D2D-enabled environments, nodes can exchange small data samples with trusted neighbors to calculate similarities locally and report them to the server.

In Fig.~\ref{fig:SimilarOffload}, if the server samples the cars with the highest computational capabilities, i.e., $A$ and $B$, the sampling is inefficient due to the high data similarity between them. Also, if it samples those with the lowest similarity, i.e., $D$ and $E$, the local models will be based on low computational capabilities, which will often result in a low accuracy (and could be catastrophic in this vehicular scenario). Optimal sampling of the cars considering both data similarities and computational capabilities is thus critical to the operation of FedL.


We take this one step further to consider how D2D offloading can lead to augmented local distributions of sampled cars.
The node sampling must consider the neighborhoods of different cars and the capability of data offloading in those neighborhoods: D2D is cheaper in terms of resource utilization among those that are close in proximity, for example. The feasible offloading topology in Fig.~\ref{fig:SimilarOffload} is represented by the data offloading graph. Given $C$'s high processing capability and data dissimilarity with neighboring cars $E$ and $D$, sampling $C$ in a D2D-optimized solution can yield a composite of all three cars' distributions. The purpose of this paper is to model these relationships for a general wireless edge network and optimize the resulting sampling and offloading configurations.

\vspace{-1.1mm}
\subsection{Outline and Summary of Contributions}

\begin{itemize}[leftmargin=4mm]

\item We formulate the joint sampling and D2D offloading optimization problem for maximizing FedL model accuracy subject to realistic network resource constraints (Sec.~\ref{s:sm}). 

\item Our theoretical analysis of the offloading subproblem for a fixed sampling strategy yields a new upper bound on the convergence of FedL under an arbitrary data sampling strategy (Sec.~\ref{s:p1}). Using this bound, we derive an efficient sequential convex optimizer for the offloading strategy.

\item We propose a novel ML-based methodology that learns the desired combination of sampling and resulting offloading (Sec.~\ref{s:p2}). We encapsulate the network structure and offloading scheme into model features and learn a mapping to the sampling strategy that maximizes expected FedL accuracy.

\item We evaluate our methodology through experiments on real-world ML tasks with network parameters from our testbed of wireless IoT devices (Sec.~\ref{s:numRes}). Our results demonstrate model accuracies that exceed FedL trained on all devices with significant reductions in processing requirements.
\end{itemize}
\vspace{-1mm}
\section{System and Optimization Model} \label{s:sm}



\noindent In this section, we formulate the joint sampling and offloading optimization (Sec.~\ref{ss:ovr_problem}). We first introduce our edge device (Sec.~\ref{sss:devices}), network (Sec.~\ref{sss:graph}), and ML (Sec.~\ref{ss:mlp}) models.


\subsection{Edge Device Model} \label{sss:devices}
We consider a set of devices $\mathcal{N} = \{1,\cdots,N\}$ connected to a server, and time span $t=0,\cdots,T$ for model training. Each device $i \in \mathcal{N}$ possesses a data processing capacity $P_i(t) \geq 0$, which limits the number of datapoints it can process for training at time $t$, and a unit data processing cost $p_i(t) \geq 0$. Intuitively, $p_i(t)$, $P_i(t)$ are related to the total CPU cycles, memory (RAM), and power available at device $i$~\cite{morabito2018legiot}. These factors are heterogeneous and time-varying, e.g., as battery power fluctuates and as each device becomes occupied with other tasks. Additionally, for each $i \in \mathcal{N}$, we define $\Psi_{i}(t) > 0$ as the data transmit budget, and $\psi_{i,j}(t) > 0$ as the unit data transmission cost to device $j$. 
Intuitively, $\psi_{i,j}(t)$, $\Psi_{i}(t)$ are dependent on factors such as the available bandwidth and the channel interference conditions. For example, devices that are closer in proximity would be expected to have lower $\psi_{i,j}(t)$.

Due to resource constraints, the server selects a set $\mathcal{S} \subseteq \mathcal{N}$ of devices to participate in FedL training. Some devices $i \in \hat{\mathcal{S}} \triangleq \mathcal{N} \setminus \mathcal{S}$ may be stragglers, i.e., possessing insufficient $P_i(t)$ to participate in training, but nonetheless gather data. Different from most works, our methodology will seek to leverage the datasets captured by nodes in the unsampled set $\hat{\mathcal{S}}$ via local D2D communications with nodes in the sampling set $\mathcal{S}$.

We denote the dataset at device $i$ for the specific ML application by $\mathcal{D}_i(t)$. $\mathcal{D}_i(0)$ is the initial data at $i$, which evolves as offloading takes place. Henceforth, we use calligraphic font (e.g., $\mathcal{D}_i(t)$) to denote a set, and non-calligraphic (e.g., $D_i(t) = |\mathcal{D}_i(t)|$) to denote its cardinality. Each data point $d \in \mathcal{D}_i(t)$ is represented as $d=(\mathbf{x}_d,y_d)$, where $\mathbf{x}_d\in \mathbb{R}^M$ is a feature vector of $M$ features, and $y_d\in\mathbb{R}$ is the target label.
\vspace{-1mm}
\subsection{Network Topology and Data Similarity Model} \label{sss:graph}
We consider a time-varying network graph $G = (\mathcal{N},\mathcal{E}(t))$, among the set of nodes $\mathcal{N}$ to represent the available D2D topology. Here, $\mathcal{E}(t)$ denotes the set of edges or connections between the nodes, where $(i,j)\in\mathcal{E}(t)$ if node $i$ is able/willing to transfer data in D2D mode to node $j$ at time $t$, depending on e.g., the trust between the devices, and whether the devices are D2D-enabled. For instance, smart home peripherals can likely transfer data to their owner's smartphone, while certain smart cars in the vehicular network in Fig.~\ref{fig:SimilarOffload} may be unwilling to communicate. 
We capture these potential D2D relationships using the adjacency matrix $\mathbf{A}(t) =[A_{i,j}(t)]_{1\leq i,j\leq N}$, where $A_{i,j}(t)=1$ if $(i,j)\in\mathcal{E}(t)$, and $A_{i,j}(t)=0$ otherwise.

We define $\Phi_{i,j}(t) \in [0,1]$ as the fraction of node $i$'s data offloaded to node $j$ at time $t$. To optimize this, we are interested in the similarity among local datasets. We define the similarity matrix $\boldsymbol{\lambda}(t) \triangleq [\lambda_{i,j}(t)]_{1\leq i,j\leq N}$ among the nodes at time $t$, where $0\leq\lambda_{i,j}(t)\leq 1$. Higher values of $\lambda_{i,j}(t)$ imply a higher dataset similarity between nodes $i$ and $j$, and thus less offloading benefit. 
In practice, neither the server nor the devices have exact knowledge of the local data distributions. To this end, we consider a probabilistic interpretation of similarity where $\lambda_{i,j}(t)$ is defined based on the probability that a data point $d_i$ sampled i.i.d from $\mathcal{D}_i(t)$ is ``similar" to at least one data point $d_j \in \mathcal{D}_j(t)$. Two datapoints $d_i$ and $d_j$ are considered  similar if (i) they have the same label $y_{d_i} = y_{d_j}$, and (ii) their feature vectors satisfy $\Vert \mathbf{x}_{d_i}-\mathbf{x}_{d_j} \Vert \leq \epsilon_0$, where $\epsilon_0 > 0$ can be application-specific. Due to device dataset heterogeneity, in general, data similarity will not be symmetric. 

For $t>0$, $\lambda_{i,j}(t)$ is defined from $\lambda_{i,j}(t-1)$ according to the offloading behavior, as we will explain in Sec.~\ref{ss:ovr_problem}. Estimates of the initial $\lambda_{i,j}(0)$, i.e., before any offloading takes place, can be obtained in several ways that avoid transferring large volumes of data. 
We assume that device $i$ will broadcast a random sample of $\mathcal{D}_i(0)$ to its neighbor $j$, which can then estimate $\lambda_{i,j}(0)$ by determining the fraction in this sample that are similar to a $d \in \mathcal{D}_j(0)$. 
To capture both the node connectivity and data similarity jointly for offloading, we also define the connectivity-similarity matrix $\boldsymbol{\Lambda}(t) = [\Lambda_{i,j}(t)]$, and $\boldsymbol{\Lambda}(t) \triangleq \boldsymbol{\lambda}(t) \circ \mathbf{A}(t)$, where $\circ$ represents the Hadamard product.


\vspace{-1mm}
\subsection{Distributed Machine Learning Model} \label{ss:mlp}
The learning objective of FedL is to train a global ML model parameterized by a vector $\mathbf{w} \in \mathbb{R}^p$ (e.g., $p$ weights in a neural network) using the devices participating in training. Formally, for $t \in \{0, \cdots, T\}$, each sampled device $i \in \mathcal{S}$ is concerned with updating its model parameter vector $\mathbf{w}_i(t)$ on its local dataset $\mathcal{D}_i(t)$. The local loss at device $i \in \mathcal{N}$ is defined as
\begin{align}
    F(\mathbf{w}_i(t)|\mathcal{D}_i(t)) = \frac {\sum_{d \in \mathcal{D}_i(t)} f(\mathbf{w}_i(t),\mathbf{x}_d,y_d)}{D_i(t)},
\end{align}
where $f(\mathbf{w}_i(t),x_d,y_d)$ denotes the corresponding loss (e.g., squared error for a regression problem, cross-entropy for a classification problem~\cite{wang2019adaptive}) of each datapoint $d \in \mathcal{D}_i$. Each device minimizes its local loss sequentially via gradient descent:
\begin{equation} \label{eq:gd}
    \mathbf{w}_i(t) = \mathbf{w}_i(t-1) - \eta \nabla F(\mathbf{w}_i(t-1)|\mathcal{D}_i(t)),
\end{equation}
where $\eta > 0$ is the step size and $\nabla F(\mathbf{w}_i(t-1)|\mathcal{D}_i(t))$ is the average gradient over the local dataset $\mathcal{D}_i(t)$. With periodicity $\tau$, the server performs a weighted average of $\mathbf{w}_i (k\tau)$, $i \in \mathcal{S}$:
\begin{equation}\label{eq:w_agg}
    \mathbf{w}_{\mathcal{S}}(k\tau) = \frac{\sum_{i \in \mathcal{S}}  \Delta_i(k\tau) \mathbf{w}_i(k\tau)}{\sum_{i \in \mathcal{S}} \Delta_i(k\tau)},
\end{equation} 
where $k$ denotes the $k$-th aggregation, $k \in \{1,\cdots,K \}$, $K = \lceil T/\tau \rceil$, and $\Delta_{i}(k\tau) \triangleq  \sum_{t=(k-1)\tau+1}^{k\tau}D_i(t)$ denotes the total data located at node $i$ between $k-1$ and $k$. The server then synchronizes the sampled devices: $\mathbf{w}_i(k\tau) \leftarrow \mathbf{w}_\mathcal{S}(k\tau)$, $\forall i \in \mathcal{S}$. 

Since we are concerned with the performance of the global parameter $\mathbf{w}_{\mathcal{S}}$, we  define $\mathbf{w}_{\mathcal{S}}(t)$ as the weighted average of $\mathbf{w}_i(t)$ as in~\eqref{eq:w_agg} for each time $t$, though it is only computed at the server when $t = k\tau$.
The global loss that we seek to optimize considers the loss over all the datapoints in the network:  
\begin{align} \label{eq:Lws}
\hspace{-4mm}
F\hspace{-0.5mm}\left(\mathbf{w}_{\mathcal{S}}(t)|\mathcal{D}_{\mathcal{N}}(t) \right)
= \frac{\sum_{i \in \mathcal{N} } D_i(t) F(\mathbf{w}_{\mathcal{S}}(t) \vert \mathcal{D}_i(t))}{D_{\mathcal{N}}(t)},
\hspace{-2mm}
\end{align}
where $\mathcal{D}_{\mathcal{N}}(t)$ denotes the multiset of the datasets of all the devices at time $t$, and $D_{\mathcal{N}}(t) \triangleq \sum_{i \in \mathcal{N}} D_i(t)$. 

\subsection{Joint Sampling and Offloading Optimization} \label{ss:ovr_problem}

The goal of our optimization is to select (i) the subset of devices $\mathcal{S}^{\star}$ to sample from a total budget of $S$ and (ii) the offloading ratios $\Phi^{\star}_{i,j}(t)$ between the devices to minimize the loss associated with $\mathbf{w}_\mathcal{S}(t)$. We consider a time average for the objective, as devices may rely on intermediate aggregations for real-time inferences. For the variables, we define the binary vector $\mathbf{x} \triangleq (x_1,\cdots,x_N)$ to represent device sampling status, i.e., if $i \in \mathcal{S}$ then $x_i = 1$, otherwise $x_i=0$, and matrix $\boldsymbol{\Phi}(t) \triangleq [\Phi_{i,j}(t)]_{1 \leq i,j \leq N}$ to represent the offloading ratios at time $t$. The resulting optimization problem $\boldsymbol{\mathcal{P}}$ is as follows:
\vspace{-4mm}

{ \small 
\begin{align}
 &\hspace{-1mm} (\boldsymbol{\mathcal{P}}):~ 
 \underset{\mathbf{x},\{\boldsymbol{\Phi}(t)\}_{t=1}^{T}}{\textrm{minimize}}~
\frac{1}{T}\sum_{t=1}^{T} F(\mathbf{w}_\mathcal{S}(t)|\mathcal{D}_{\mathcal{N}}(t)) \label{eq:obj}\hspace{-1mm} \\
&\hspace{-1mm}  \textrm{subject to}\nonumber\\ 
&\hspace{-1mm}  D_i(t) = D_i(t-1) + R_{i}(t), ~i\in\mathcal{N}, \label{eq:con1} \\
&\hspace{-1mm}  R_{i}(t) = \sum_{k \in \mathcal{N}} D_k(t-1)\Phi_{k,i}(t)(1-{\Lambda}_{k,i}(t-1)),~i\in\mathcal{N}, \label{eq:con2} \\ 
&\hspace{-1mm}  \Lambda_{k,i}(t) = \Lambda_{k,i}(t-1) + (1-\Lambda_{k,i}(t-1))\Phi_{k,i}(t),~ i,k\in\mathcal{N}, \label{eq:con3} \\
&\hspace{-1mm}  R_{i}(t) \leq \theta_i(t), ~i\in\mathcal{N}, \label{eq:con4}\\
&\hspace{-1mm}  p_i(t) D_i(t) \leq P_i(t), ~i\in\mathcal{N}, \label{eq:con5}\\
&\hspace{-1mm}  D_k(t-1) \sum_{i\in \mathcal{N}} x_i\Phi_{k,i}(t)\psi_{k,i}(t) \leq \Psi_{k}(t), ~k\in\mathcal{N}, \label{eq:con6}\\
&\hspace{-1mm} \sum_{i \in \mathcal{N}} \Phi_{k,i}(t) \leq 1, ~k\in\mathcal{N}, \label{eq:con7} \\ 
&\hspace{-1mm} (1-x_i)(1-x_k)\Phi_{k,i}(t) = 0, ~i,k \in \mathcal{N}, \label{eq:con8}\\ 
&\hspace{-1mm} x_k\Phi_{k,i}(t) = 0, ~i,k\in \mathcal{N}, \label{eq:con9}\\ 
&\hspace{-1mm} (1-A_{k,i}(t)) \Phi_{k,i}(t) = 0, ~i,k \in \mathcal{N}, \label{eq:con10}\\
&\hspace{-1mm} \sum_{i\in \mathcal{N}} x_i=S,\label{eq:con11}\\
&\hspace{-1.1mm}  {\Lambda}_{k,i}(t) \geq 0, \Lambda_{k,i}(t) \leq 1, \Phi_{k,i}(t) \geq 0,~x_i \in \{0,1\},~i,k\in\mathcal{N}\hspace{-.2mm}. \label{eq:con12}\hspace{-3mm}
\end{align} }

\vspace{-0.22in}
\noindent The data at sampled devices, i.e., $D_i(t)$ for $i \in \mathcal{S}$, changes over time in~\eqref{eq:con1} based on the total received data $R_i(t)$ for device $i$. $R_i(t)$ is determined in~\eqref{eq:con2} by scaling the data transmissions from $k \in \hat{\mathcal{S}}$ to device $i$ according to the similarity. In response to the data offloading, the connectivity-similarity matrix is updated in~\eqref{eq:con3}. Together,~\eqref{eq:con2} and~\eqref{eq:con3} capture similarity-aware D2D offloading, which we explain further in the following paragraph. Next,~\eqref{eq:con4}-\eqref{eq:con6} ensure that our D2D offloading solution adheres to device receive capacities $\theta_i(t)$, data processing limits $P_i(t)$, and D2D limits $\Psi_i(t)$. \eqref{eq:con7} ensures total offloaded data by device $k$ does not exceed its local dataset size. Through~\eqref{eq:con8}-\eqref{eq:con10}, offloading only occurs between single-hop D2D neighbors from $k \in \hat{\mathcal{S}}$ to $i \in \mathcal{S}$. \eqref{eq:con11} maintains compliance with the desired sampling size, i.e., $|\mathcal{S}^{\star}| = S$. 


\textbf{Similarity-aware D2D offloading:} The amount of raw data device $i$ receives from $k$ is $\Phi_{k,i}(t) D_{k}(t-1)$. 
Ideally, device $i$ will receive data that is dissimilar to $\mathcal{D}_i(0)$. 
However, neither $i$ nor $k$ have full knowledge of each others' datasets in this distributed scenario (nor does the server). Therefore, data offloading in $\boldsymbol{\mathcal{P}}$ is conducted through an i.i.d. selection of $\Phi_{k,i}(t) D_{k}(t-1)$ data points from $k$ to send to $i$. The estimated overlapping data that arrives at $i$ is $\Phi_{k,i}(t) D_{k}(t-1) \Lambda_{k,i}(t-1)$, and the resulting useful data is $\Phi_{k,i}(t) D_{k}(t-1)(1-\Lambda_{k,i}(t-1))$. 
We therefore adjust $\Lambda_{k,i}(t)$ by the effective fraction of data $\Phi_{k,i}(t) (1-\Lambda_{k,i}(t-1))$ offloaded from $k$ to $i$, resulting in \eqref{eq:con3}. In particular, when $k$ transfers all of its data to $i$ (i.e., $\Phi_{k,i}(t)=1$), $\Lambda_{k,i}(t)$ becomes 1, preventing further data offloading from $k$ to $i$ according to~\eqref{eq:con2}. Imposing these constraints promotes data diversity among the sampled nodes through offloading. 

\textbf{Solution overview:} Problem~$\boldsymbol{\mathcal{P}}$ faces two major challenges: (i) it requires a concrete model of the loss function with respect to the datasets, which is, in general, intractable for deep learning models~\cite{goodfellow2016deep}, and (ii) even if the loss function is known, the coupled sampling and offloading procedures make this problem an NP-hard mixed integer programming problem. To overcome this, we will first consider the offloading subproblem for a fixed sampling strategy, and develop a sequential convex programming method to solve it in Sec.~\ref{s:p1}. Then, we will integrate this into a graph convolutional network (GCN)-based methodology that learns the relationship between the network properties, sampling strategy (with its corresponding offloading), and the resulting FedL model accuracy in Sec.~\ref{s:p2}. An overall flowchart of our methodology is given in Fig.~\ref{fig:BigPicture}.

\vspace{-1mm}
\section{Developing the Offloading Optimizer} 
\label{s:p1}
\noindent In this section, we study the offloading optimization subproblem of $\boldsymbol{\mathcal{P}}$. Our theoretical analysis of~\eqref{eq:obj} under common assumptions will yield an efficient approximation of the objective in terms of the offloading variables (Sec.~\ref{ss:p1b}). We will then develop our solver for the resulting optimization (Sec.~\ref{ss:p1c}). 

\vspace{-1mm}
\subsection{Definitions and Assumptions} \label{ss:p1def}
To aid our theoretical analysis of FedL, similar to \cite{wang2019adaptive}, we will consider a hypothetical ML training process that has access to the entire dataset $\mathcal{D}_{\mathcal{N}}(t)$ at each time instance. The parameter vector $\mathbf{v}_k(t)$ for this centralized model is trained as follows: (i) at each global aggregation $t = k\tau$, $\mathbf{v}_k(t)$ is synchronized with $\mathbf{w}_{\mathcal{S}}(t)$, i.e., $\mathbf{v}_k(t) \leftarrow \mathbf{w}_{\mathcal{S}}(t)$, and (ii) in-between global aggregation periods, $\mathbf{v}_k(t)$ is trained based on gradient descent iterations to minimize the global loss $F(\mathbf{v}_k(t) \vert \mathcal{D}_{\mathcal{N}}(t))$. 
\begin{definition} [Difference between sampled and unsampled gradients] \label{def:e} 
We define the instantaneous difference between $\nabla F(\mathbf{w}_{\mathcal{S}}(t)\vert \mathcal{D}_{\mathcal{N}}(t))$, the gradient with respect to the full dataset across the network, and $\nabla F(\mathbf{w}_{\mathcal{S}}(t)\vert \mathcal{D}_{\mathcal{S}}(t))$, the gradient with respect to the sampled dataset, as:
\begin{equation}
\hspace{-0.0mm}
\begin{aligned}  \label{eq:l1_result} 
&\zeta(\mathbf{w}_{\mathcal{S}}(t)) \triangleq \frac{G_{\mathcal{S}}(t)}{D_{\mathcal{S}}(t)} - \frac{\sum_{i \in \mathcal{N}} D_i(t) \nabla F(\mathbf{w}_{\mathcal{S}}(t)|\mathcal{D}_i(t))}{D_{\mathcal{N}}(t)},
\end{aligned}
\hspace{-4mm}
\end{equation} 
where~$G_{\mathcal{S}}(t) \triangleq \sum_{i \in \mathcal{S}}D_i(t)\nabla F(\mathbf{w}_{\mathcal{S}}(t)|\mathcal{D}_i(t))$ is the scaled sum of gradients on the sampled datasets, and~$D_{\mathcal{S}}(t) = \sum_{i \in \mathcal{S}}D_i(t)$ is the total data across the sampled devices.
\end{definition}


\begin{definition}[Difference between sampled and unsampled gradients] We define $\delta_i(t)$ as the upper bound between the gradient computed on $\mathcal{D}_i(t)$ for $i \in \mathcal{S}$ and $\mathcal{D}_{\mathcal{N}}(t)$ at time $t$: 
\label{def:deltai}

\vspace{-0.16in}
\small
\begin{equation}
\hspace{-.01mm}
     \Vert \nabla F(\mathbf{w}_{\mathcal{S}}(t)|\mathcal{D}_i(t)) \hspace{-.45mm}- \hspace{-.45mm} \nabla F(\mathbf{w}_{\mathcal{S}}(t)|\mathcal{D}_{\mathcal{N}}(t))\hspace{-.45mm} -\hspace{-.45mm} \zeta(\mathbf{w}_{\mathcal{S}}(t)) \hspace{-.4mm}\Vert\hspace{-.5mm}\leq\hspace{-.5mm} \delta_i(t).
    \hspace{-14mm}
\end{equation}
\end{definition}


\vspace{-0.05in}
We also make the following standard assumptions~\cite{wang2019adaptive,tu2020network} on the loss function $F(\mathbf{w})$ for the ML model being trained: 
\begin{assumption} 
We assume $F(\mathbf{w})$ is convex with respect to $\mathbf{w}$, L-Lipschitz, i.e., $\Vert F(\mathbf{w}) - F(\mathbf{w}^{\prime}) \Vert \leq L\Vert \mathbf{w} - \mathbf{w}^{\prime}\Vert $, and  $\beta$-smooth, i.e., $\Vert \nabla F(\mathbf{w}) - \nabla F(\mathbf{w}^{\prime}) \Vert \leq \beta \Vert \mathbf{w} - \mathbf{w}^{\prime}\Vert$, $\forall \mathbf{w}, \mathbf{w}^{\prime}$.
\end{assumption}
Despite these assumptions, we will show in Sec.~\ref{s:numRes} that our results still obtain significant improvements in practice for neural networks which do not obey the above assumptions.
\vspace{-1mm}
\subsection{Upper Bound on Convergence} \label{ss:p1b}
For convergence analysis, we assume that devices only offload the same data once, and assume that recipient nodes always keep received data. This must be done to ensure that the optimal global model parameters remain constant throughout time.
The following theorem gives an upper bound on the difference between the parameters of sampled FedL and those from the centralized learning, i.e., $\Vert \mathbf{w}_{\mathcal{S}}(t) - \mathbf{v}_k(t) \Vert$, over time:

\begin{theorem}[Upper bound on the difference between sampled FedL and centralized learning]\label{thm:error} 
Assuming $\eta \leq {\beta}^{-1}$, the upper-bound on the difference between $\mathbf{w}_{\mathcal{S}}(t)$ and $\mathbf{v}_k(t)$ within the local update period before the $k$-th global aggregation, $t \in \{(k-1)\tau+1,...,k\tau\}$, is given by:
\vspace{-4.5mm}

\small
\begin{equation} \label{th1:1}
\hspace{-0mm}
    \Vert \mathbf{w}_{\mathcal{S}}(t) - \mathbf{v}_{k}(t)\Vert \leq \frac{1}{\beta} \hspace{-5mm} \sum_{y = (k-1)\tau+1}^{t} \hspace{-1mm} \bigg(\Upsilon (y,k) + \Vert \zeta(\mathbf{w}_{\mathcal{S}}(y-1))\Vert\bigg),
    \hspace{-3mm}
\end{equation}
\normalsize
where $\Upsilon (y,k) \triangleq \delta_{\mathcal{S}}(y) (2^{y-1-(k-1)\tau}-1)$,
and
\begin{equation} \label{th1:3}
   \delta_{\mathcal{S}}(t) \triangleq \left({\sum_{i \in \mathcal{S}}D_i(t)\delta_i(t)}\right)\left({\sum_{i \in \mathcal{S}}D_i(t)}\right)^{-1}.
\end{equation} 
\end{theorem}

\begin{proof}
See Appendix~\ref{app:main1}.
\end{proof}
Through $\delta_{\mathcal{S}}(t)$, Theorem~\ref{thm:error} establishes a relationship between the difference in model parameters and the datapoints $D_i(t)$ in the sampled set $i \in \mathcal{S}$. Using this, we obtain an upper bound on the difference between our $\mathbf{w}_{\mathcal{S}}(t)$ and the global minimizer of model loss $\mathbf{w}^*(t) =\argmin_{\mathbf{w}}   F(\mathbf{w}|\mathcal{D}_{\mathcal{N}}(t))$:

\begin{corollary}[Upper bound on the difference between sampled FedL and the optimal] \label{c1}
The difference of the loss induced by $\mathbf{w}_{\mathcal{S}}(t)$  compared to the loss induced by $\mathbf{w}^*(t)$ for $t \in \{(k-1)\tau,\cdots,k\tau-1\}$, is given by:
\begin{equation} \label{eq:cl1_result}
\begin{aligned}
& F(\mathbf{w}_{\mathcal{S}}(t)|\mathcal{D}_{\mathcal{N}}(t)) - F(\mathbf{w}^*(t)|\mathcal{D}_{\mathcal{N}}(t))\leq  \\
&g(\hat{\Upsilon}(\hat{K})) \triangleq \left(t \xi \eta \left(1 - \frac{\beta \eta}{2}\right) - \frac{(\hat{K}+1)L}{\beta \epsilon^2} \hat{\Upsilon}(\hat{K}) \right)^{-1} ,
\end{aligned}
\end{equation}
\normalsize
where $\hat{\Upsilon}(\hat{K})\hspace{-0.5mm} \triangleq\hspace{-0.5mm} \sum_{y=(\hat{K}-1)\tau+1}^{\hat{K}\tau}  \left(\Upsilon(\hat{K},y) \hspace{-0.5mm}+\hspace{-0.5mm} \Vert \zeta(\mathbf{w}_{\mathcal{S}}(y\hspace{-0.5mm}-\hspace{-0.5mm}1))\Vert \right)$, $\hat{K} = \floor{t /\tau}$, and $\xi = \min_k \frac{1}{{\Vert \mathbf{v}_k((k-1)\tau) - \mathbf{w}^*(t) \Vert}^2}$.
\end{corollary}

\begin{proof}
See Appendix~\ref{app:c1}.
\end{proof}
As our ultimate goal is an expression of~\eqref{eq:obj} in terms of the data $\mathcal{D}_i(t)$ at each node, we consider the relationship between $g(\hat{\Upsilon}({\hat{K}}))$ and $\mathcal{D}_i(t)$, which is clearly non-convex through $\hat{\Upsilon}({\hat{K}})$. Since $\hat{\Upsilon}(\hat{K}) \ll 1$ (see Appendix~\ref{app:c1}),~\eqref{eq:cl1_result} can be approximated using the first two terms of its Taylor series:
\begin{equation} \label{eq:fupsilon}
\begin{aligned}
& g(\hat{\Upsilon}) \approx \frac{1}{t \xi \eta (1 -\frac{\eta \beta}{2})} + \frac{{(\hat{K}+1)L}}{\beta \epsilon^2\left(t \xi \eta (1 -\frac{\eta \beta}{2})\right)^2} \hat{\Upsilon}.
\end{aligned}
\end{equation} 
At each time instant, the first term in the right hand side (RHS) of \eqref{eq:fupsilon} is a constant. Thus, under this approximation, the RHS of \eqref{eq:cl1_result} becomes proportional to $\hat{\Upsilon}$, which is in turn a function of $\delta_{\mathcal{S}}(t)$. The final step is to bound the expression for $\delta_i(t)$, and thus their weighted sum $\delta_{\mathcal{S}}(t)$, in terms of the $D_i(t)$, $\forall i \in \mathcal{S}$.



\begin{proposition}[Upper bound on the difference between local gradients] \label{prop:1} The difference in gradient with respect to a sampled device dataset vs. the full dataset satisfies:
\begin{equation}\label{eq:lemma}
\hspace{-0mm}
\begin{aligned} 
&\Vert\nabla F(\mathbf{w}_{\mathcal{S}}(t)|\mathcal{D}_i(t)) - \nabla F\left(\mathbf{w}_{\mathcal{S}}(t)|\mathcal{D}_{\mathcal{N}}(t) \right) - \zeta(\mathbf{w}_{\mathcal{S}}(t))\Vert \\
&  \leq \left(\frac{D_{\mathcal{N}}(t)-D_{\mathcal{S}}(t)}{D_{\mathcal{N}}(t)}\right) \overline{\nabla F(t)} + \frac{\gamma}{\sqrt{D_i(t)}} + C \equiv \delta_i(t), 
\end{aligned}
\hspace{-5mm}
\end{equation}
where $C \triangleq \left({D_{\mathcal{N}}(t)}\right)^{-1}\sum_{i \in \hat{\mathcal{S}}} D_i(t)\nabla F(\mathbf{w}_{\mathcal{S}}(t)|\mathcal{D}_i(t))$, $\hat{\mathcal{S}} = \mathcal{N} \setminus \mathcal{S}$, $\gamma$ is a constant independent of $\mathcal{D}_i(t)$, and 
\begin{equation} \label{eq:nablaFbar} 
\overline{\nabla F(t)} \triangleq \left({D_{\mathcal{S}}(t)}\right)^{-1}{\sum_{i \in \mathcal{S}}D_i(t)\nabla F(\mathbf{w}_{\mathcal{S}}(t)|\mathcal{D}_i(t))}.
\end{equation}
\end{proposition}
\begin{proof}
See Appendix~\ref{app:prop1}.
\end{proof}
\normalsize
The above proposition relates each $\delta_i(t)$ to the number of instantaneous data points available at device $i$. 

\subsection{Offloading as a Sequential Convex Optimization} \label{ss:p1c}

Using the result of~\eqref{eq:fupsilon} to replace the RHS of~\eqref{eq:cl1_result} implies that the objective function in~\eqref{eq:obj} is proportional to $\frac{1}{T}\sum_{t=1}^{T} \delta_{\mathcal{S}}(t)$, where $\delta_{\mathcal{S}}(t)$ is defined in~\eqref{th1:3} as a sum-of-ratios of $\delta_i(t)$. Considering $\frac{1}{T}\sum_{t=1}^{T} \delta_{\mathcal{S}}(t)$ as the objective in problem $\boldsymbol{\mathcal{P}}$ yields the sum-of-ratios problem in fractional programming~\cite{schaible2003fractional}. 
The scale of existing solvers for the sum-of-ratios fractional programming problem (e.g. \cite{kuno2002branch}) 
are on the order of ten ratios, which corresponds to ten devices in our case. Contemporary large-scale networks that may have hundreds of edge devices~\cite{8373692} therefore cannot be solved accurately or in a time-sensitive manner. Motivated by the above fact, we approximate $\delta_{\mathcal{S}}(t) \approx \frac{1}{S} \sum_{i \in \mathcal{S}} \delta_i(t)$. Using this with~\eqref{eq:lemma}, we obtain the following approximation for~\eqref{eq:obj}:
\vspace{-4mm}

\small
\begin{align} \label{eq:obj_temp}
& \frac{1}{T}\sum_{t=1}^{T} \underbrace{\left(\frac{D_{\mathcal{N}}(t)-D_{\mathcal{S}}(t)}{D_{\mathcal{N}}(t)}\right)  \overline{\nabla F(t)}}_{(a)} + \frac{1}{\vert\mathcal{S}\vert} \sum_{i \in \mathcal{S}}\underbrace{ \frac{\gamma}{\sqrt{D_i(t)}}}_{(b)},
\end{align}
\normalsize
where term $(a)$ is due to sampling and term $(b)$ is the statistical error from the central limit theorem. Thus, for a known binary vector $\mathbf{x}$ (i.e., a known $\mathcal{S}$)  that satisfies (16), we arrive at the following optimization problem for the D2D data offloading:
\begin{align}
 &(\boldsymbol{\mathcal{P}}_D):~~ 
 \underset{\{\boldsymbol{\Phi}(t)\}_{t=1}^{T}}{\min} \eqref{eq:obj_temp}
&  \textrm{s.t.}~ (6)-(15), (17).\nonumber
\end{align} 
Since the number of datapoints at the unsampled devices is fixed for all time, $D_{\mathcal{N}}(t)$ can be expressed as $D_{\mathcal{S}}(t) + D_{\hat{\mathcal{S}}}$, where $D_{\hat{\mathcal{S}}} = \sum_{i \in \hat{\mathcal{S}}} D_i$ is a constant. 
Consequently, both the coefficient of $\overline{\nabla F(t)}$ in term $(a)$ and the entirety of term $(b)$ in \eqref{eq:obj_temp} are decreasing functions of the quantity of data $D_i(t)$ at sampled devices $i \in \mathcal{S}$. Furthermore, given $\overline{\nabla F(t)}$, through~\eqref{eq:con1} and~\eqref{eq:con2}, both terms $(a)$ and $(b)$ are convex with respect to the offloading variables in Problem $\boldsymbol{\mathcal{P}}_D$. 
The only remaining challenge is then to obtain $\overline{\nabla F(t)}$, which we consider next. 




\begin{figure}[t]
\includegraphics[width=.47\textwidth]{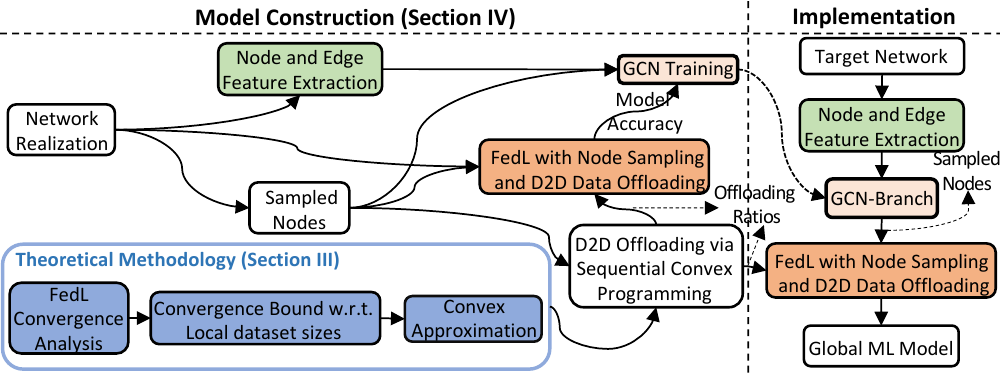}
\centering
\caption{Overview of the joint sampling and offloading methodology developed in Sec.~\ref{s:p1}\&\ref{s:p2}. During model construction, the offloading optimizer from Sec.~\ref{s:p1} is used to determine offloading for a set of sampled devices. The GCN-based algorithm developed in Sec.~\ref{s:p2} determines the combination of sampling and optimized offloading expected to maximize FedL accuracy. Then, the resulting model is applied to the target network for FedL implementation.}
\label{fig:BigPicture}
\vspace{-5mm}
\end{figure}

\textbf{Sequential gradient approximation:}
Obtaining $\overline{\nabla F(t)}$ requires the knowledge of real-time gradients, $\nabla F(\mathbf{w}_{\mathcal{S}}(t)|\mathcal{D}_i(t))$, $\forall i \in \mathcal{S}$, which are unknown a priori. Furthermore, the gradients of the devices are only observed at the global aggregation time instances $t = k\tau$. Motivated by this, we approximate $\overline{\nabla{{F}(t)}}$ for $t \in \{k\tau+1,\cdots,(k+1)\tau\}$, $ k \in \{1,\cdots,K\}$, using the gradients observed at the most recent global aggregation, i.e., $\nabla F(\mathbf{w}_{\mathcal{S}}(k\tau)|\mathcal{D}_i(k\tau))$, $i\in \mathcal{S}$ on which we perform a sequence of corrective approximations.
Specifically, since the average loss $F$ is convex, $\overline{\nabla F(t)}$ is expected to decrease over time. We assume that this decrease occurs linearly and approximate the real-time gradient using the previously observed gradient at the server as $\overline{\nabla F(t)} \approx \overline{\nabla F(k\tau)} / {\alpha^{t-k\tau}_{k+1}}$, $t\in \{k\tau+1,\cdots,(k+1)\tau\}$, $\forall k \in \{1,\cdots,K\}$, where the scaling factor $\alpha_{k+1} \hspace{-.5mm}> \hspace{-.5mm} 1$ is re-adjusted after every global aggregation $k$. Through the re-adjustment procedure, the server receives the gradients and computes the scaling factor for the each aggregation period as $\alpha_{k+1} = \sqrt[\tau]{\overline{\nabla F((k-1)\tau)}/ \overline{\nabla F(k\tau)}}$. 

Given the aforementioned characteristics of terms $(a)$ and $(b)$ in~\eqref{eq:obj_temp}, our proposed iterative approximation of $\overline {\nabla F(t)}$, and the fact that the constraints of $\boldsymbol{\mathcal{P}}_D$ are all affine at each time instance, we can solve this problem as a sequence of convex optimization problems over time. For this, we employ the CVXPY convex optimization software \cite{diamond2016cvxpy}. 



\section{Smart Device Sampling with D2D Offloading} \label{s:p2}
\begin{figure} 
    \centering
    \includegraphics[width=0.48\textwidth]{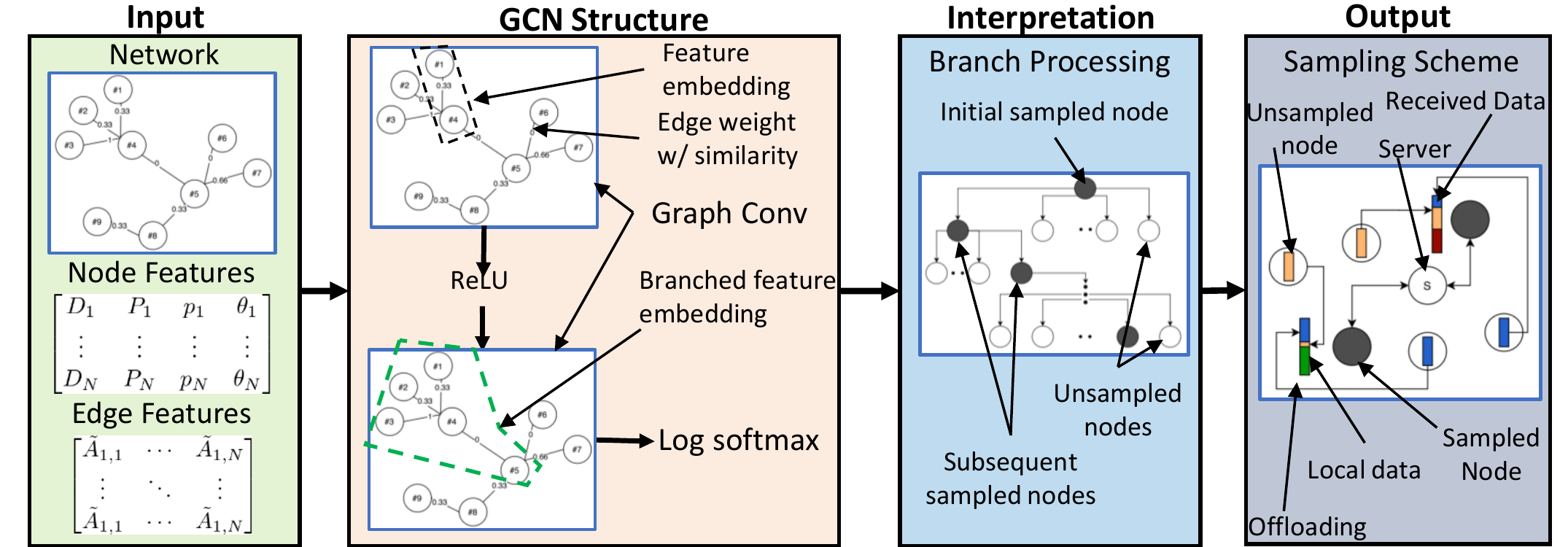}
    \caption{Architecture of our GCN-branch sampling algorithm. 
    Each GCN layer convolves node features in local neighborhoods from the input network. 
    The output probability vector informs the choice of sampled nodes $\mathcal{S}$. 
    This result is then passed to the offloading optimizer to determine offloading between sampled and unsampled nodes.}
    \label{fig:gcn_architecture}
    \vspace{-6mm}
\end{figure}

\noindent 
We now turn to the sampling decisions in problem $\boldsymbol{\mathcal{P}}$, which must be coupled with the offloading solution to $\boldsymbol{\mathcal{P}}_{D}$. After explaining the rationale for our GCN-based approach (Sec.~\ref{ss:sp2_intro}), we will detail our training procedure encoding the network characteristics (Sec.~\ref{ss:gcnmodel}). Finally, we will develop an iterative procedure for selecting the sample set (Sec.~\ref{ss:fullimplement}).

\subsection{Rationale and Overview of GCN Sampling Approach} \label{ss:sp2_intro}
Sampling the optimal subset of nodes from a resource-constrained network to maximize a utility function (in our case, minimizing the ML loss) has some similarity to  0-1 knapsack problem~\cite{martello2000new}. 
In this combinatorial optimization problem, a set of weights and values for $n$ items are given, where each item can be either added or left out to maximize the value of the items within the knapsack subject to a weight capacity. Analogously, our sampling problem aims to maximize FedL accuracy 
while adhering to a sampling budget $S = \sum_{i \in \mathcal{N}}x_i$.
Strategies for the knapsack problem become unsuitable here because the value that each device provides to FedL is difficult to quantify: it depends on the ML loss function, the gradient descent procedure, and the D2D relationships from Sec.~\ref{s:p1}.

To address these complexities, we propose a (separate) ML technique to model the relationship between network characteristics, the sampling set, and the resulting FedL model quality. Specifically, we develop a sampling technique based on active filtering of a Graph Convolutional Network (GCN)~\cite{kipf2016semi,lee2020fast}. In a GCN, the learning procedure consists of sequentially feeding an input (the network graph) through a series of graph convolution~\cite{wu2020comprehensive} layers, which generalize the traditional convolution operation into non-Euclidean spaces and thus captures connections among nodes in different neighborhoods.

Our methodology is depicted in Fig.~\ref{fig:gcn_architecture}. GCNs excel at graph-based classification tasks, as they learn over the intrinsic graph structure. However, GCNs by themselves have performance issues when there are multiple good candidates for the classification problem~\cite{li2018combinatorial}. This holds for our large-scale network scenario, as many high performing sets of sampling candidates can be expected. The data offloading scheme adds another important dimension: a sampled node $i$ may perform poorly when considered in isolation, but it may have high processing capacities $P_i(t)$ and be connected to unsampled nodes $j \in \hat{\mathcal{S}}$ with large quantities of local data and high transfer limits $\Psi_{j,i}(t)$. We address these issues by (i) incorporating the solution from Sec.~\ref{s:p1} into the GCN training procedure, and (ii) proposing \textit{sampling GCN-branch}, a network-based post-processing technique that maps the GCN output to a sampling set by considering the underlying connectivity-similarity matrix.

\begin{figure}[t] 
\centering
\includegraphics[width=60mm,height=25mm]{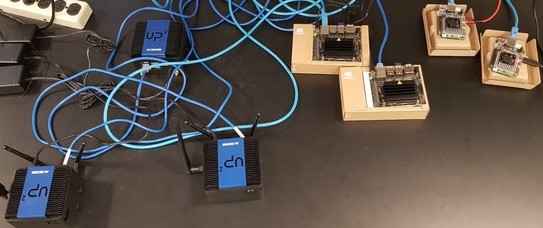}
\caption{IoT testbed used to generate device and link characteristics.}
\vspace{-6mm}
\label{fig:testbed}
\end{figure} 

\subsection{GCN Architecture and Training Procedure} \label{ss:gcnmodel}

We consider a GCN function $H(\boldsymbol{\pi},\tilde{\mathbf{A}})$ with two inputs: (i)~$\boldsymbol{\pi} \in \mathbb{R}^{N \times U}$, a matrix of $U$ node features, and (ii) $\mathbf{\tilde{A}} \in \mathbb{R}^{N \times N}$, the augmented connectivity-similarity matrix. The feature vector for each node $i$ is defined as  $\boldsymbol{\pi}_i \triangleq [D_i(0), P_i(0), p_i(0), \theta_i(0)]$, forming the rows of $\boldsymbol{\pi}$, and the augmented connectivity-similarity matrix is defined as $\tilde{\mathbf{A}}\triangleq \boldsymbol{{\Lambda}(0)}+\mathbf{I}_{N}$, where $\mathbf{I}_{N}$ denotes the identity matrix~\cite{wu2020comprehensive}. 
$H$ consists of two graph convolutional layers separated by a rectified linear unit (ReLU) activation layer~\cite{kipf2016semi}, as depicted in Fig.~\ref{fig:gcn_architecture}. The outputs of each layer are defined as:\vspace{-1mm}
\begin{equation}
    \mathbf{H}^{(l)} \triangleq \sigma\left( \tilde{\mathbf{D}}^{-\frac{1}{2}} \tilde{\mathbf{A}} \tilde{\mathbf{D}}^{-\frac{1}{2}} \mathbf{H}^{(l-1)} \mathbf{Q}^{(l)}\right),~l \in \{1,2\},
    \vspace{-1mm}
\end{equation}
where $\tilde{\mathbf{D}}$ is the degree matrix of $\tilde{\mathbf{A}}$, $\mathbf{Q}^{(l)}$ denotes the trainable weights for the $l$-th layer, and $\sigma$ represents ReLU activation. Note that $\mathbf{H}^{(0)} = \boldsymbol{\pi}$, $\mathbf{Q}^{(1)} \in \mathbb{R}^{U \times O}$, and $\mathbf{Q}^{(2)}\in \mathbb{R}^{O \times 1}$, where $O$ is the dimension of the second layer. Finally, log-softmax activation is applied to $\mathbf{H}^{(2)} \in \mathbb{R}^{N}$ to convert the results into a vector of probabilities, i.e., $\boldsymbol{\Gamma} \in [0,1]^{N}$, representing the likelihood of each node belonging to the sampled set.

\textbf{GCN training procedure:} To train the GCN weights, we generate a set of sample network and node data realizations $e = 1,\cdots,E$ with the properties from Sec.~\ref{sss:devices}\&\ref{sss:graph}. For each realization, we calculate the matrices $\boldsymbol{\pi}_e$ and $\mathbf{\tilde{A}}_e$ corresponding to the inputs of the GCN. Then, for each candidate sampling allocation $\mathbf{x}^s_e=[({x}^s_e)_i]_{1\leq i\leq N}$ (with $\sum_i ({x}^s_e)_i = S$), we solve $\boldsymbol{\mathcal{P}}_{D}$ from Sec.~\ref{s:p1} to obtain the offloading scheme, and then determine the loss of FedL resulting from model training and D2D offloading. Among these, we choose the $\mathbf{x}^{\star}_e$ that yields the smallest objective to be the target GCN output. The collection of $[(\boldsymbol{\pi}_e, \boldsymbol{\tilde{A}}_e, \mathbf{x}^{\star}_e)]_{e=1}^{E}$ form the training samples for the GCN.


As the number of devices $N$ increases, the number of choices that will be considered for the sampled set increases combinatorially as $\binom{N}{S}$. An advantage of this GCN procedure is that it is network-size independent: once trained on a set of realizations for tolerable-sized values of $N$, the graph convolutional layer weights $\mathbf{Q}^{(l)}$, $l \in \{1,2\}$, can be applied to the desired network of arbitrary size $N$. Our obtained performance results in Sec.~\ref{s:numRes} verify this experimentally.

\begin{figure}[t]
\centering
\includegraphics[width = 0.45\textwidth]{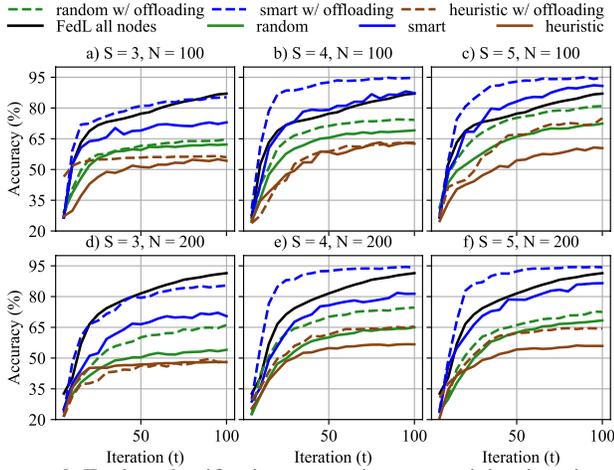}
\vspace{-2.5mm}
\caption{Testing classification accuracies over training iterations on MNIST obtained by the sampling schemes with and without offloading, and by FedL using all nodes, for different sampled sizes ($S$) and nodes ($N$). For $S > 3$, our smart sampling with offloading method consistently obtains a wide margin of improvement over all schemes.}
\label{fig:345_mnist}
\vspace{-5mm}
\end{figure}

\subsection{Offloading-Aware Smart Device Sampling} \label{ss:fullimplement} 

Given any network graph, our procedure must solve the sampling problem at the point of FedL initialization, i.e., $t=0$.
With the trained GCN in hand, we obtain $\boldsymbol{\pi}$ and $\tilde{\mathbf{A}}$ for the target network and calculate $\boldsymbol{\Gamma} = H(\boldsymbol{\pi},\tilde{\mathbf{A}})$, $\boldsymbol{\Gamma}=[\Gamma_i]_{1\leq i\leq N}$. Given this output, our sampling GCN-branch 
algorithm populates the set $\mathcal{S}$ as follows. Let $\mathcal{N}_p \subset \mathcal{N}$ be the subset of nodes in the 98th percentile of initial data quantity. Starting with $\mathcal{S} = \emptyset$, the first node is added according to $S = S \cup \{s_1\}$, where $s_1 = \argmax_{i \in \mathcal{N}_p} \Gamma_i$. In this way, the first node added is the device with highest GCN probability among the largest data generation nodes. To choose subsequent sampled nodes, the algorithm performs a recursive branch-based search on the initial connectivity-similarity matrix $\boldsymbol{\Lambda}(0)$ for nodes with the highest sampling probabilities and the least aggregate data similarity to the previously sampled nodes. 
Formally, we choose the $n$-th node addition as $\mathcal{S} = \mathcal{S} \cup \{s_n\}$, where $s_{n} =\argmax_{i \in \mathcal{R}_{s_{n-1}}} \Gamma_i$, with $\mathcal{R}_{s_{n-1}}$ denoting the neighbor nodes of $s_{n-1}$ within the 98th percentile of data dissimilarity to $s_{n-1}$ (i.e., based on $\boldsymbol{\Lambda}(0)$).
In this way, our branch algorithm relies on the GCN to decide which branch the sampling scheme will follow given its current sampled nodes (visualization in Fig.~\ref{fig:gcn_architecture}), 
so that subsequent selections are more likely to contain nodes with (i) different 
data distributions while (ii) leading to new neighborhoods that can contribute to the current set. Once the sampled set $\mathcal{S}$ is determined, the offloading is scheduled for $t = 0,...,T$ per the solver for $\boldsymbol{\mathcal{P}}_{D}$ from Sec.~\ref{s:p1}.

\textbf{Summary of methodology:} Fig.~\ref{fig:BigPicture} summarizes our methodology developed in Sec.~\ref{s:p1}\&\ref{s:p2} for solving $\boldsymbol{\mathcal{P}}$. The sequential convex optimization for offloading (Sec.~\ref{s:p1}) is embedded within the GCN-based sampling procedure (Sec.~\ref{s:p2}). Once the model is trained on sample network realizations, it is applied to the target network to generate the $\mathcal{S}$ and $\boldsymbol{\Phi}(t)$ for FedL.


\vspace{-1mm}
\section{Experimental Evaluation}\label{s:numRes}

\begin{figure}[t]
\centering
\includegraphics[width=0.45\textwidth]{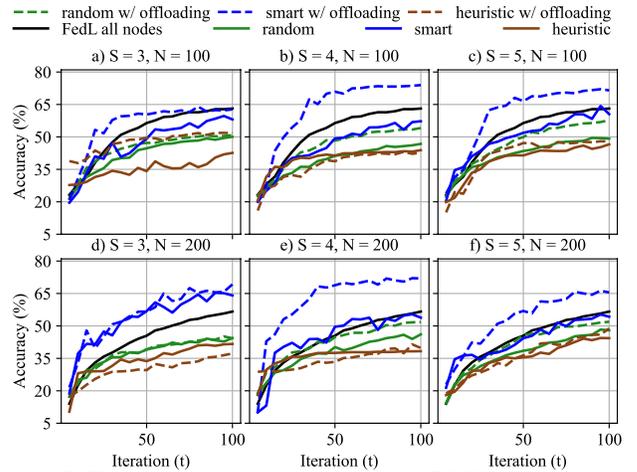}
\vspace{-2.5mm}
\caption{Testing classification accuracies on F-MNIST for the same setup as in Fig.~\ref{fig:345_mnist}. The results are consistent with the MNIST dataset. The wide margin of improvement obtained by smart sampling with offloading vs. without emphasizes the benefit of considering these two aspects jointly for FedL optimization.}
\label{fig:345_fmnist}
\vspace{-5mm}
\end{figure}


\subsection{Setup and Experimental Procedure} \label{ss:setup}
\subsubsection{Network characteristics via wireless testbed}
We employed our IoT testbed in Fig.~\ref{fig:testbed} to obtain device and communication characteristics. It consists of Jetson Nano, Coral Dev Board TPU, and UP Squared AI Edge boards configured to operate in D2D mode.
We used Dstat~\cite{dstat} to collect the device resources and power consumption. We map the measured computing resources (in CPU cycles and RAM) and the corresponding power consumption (in mW) at devices to the costs and capacities in our model by calculating the Gateway Performance Efficiency Factor (GPEF)~\cite{morabito2018legiot}. Specifically, to determine the processing costs $p_i(t)$, we measured the GPEF of the devices running gradient iterations on the MNIST dataset~\cite{yann}. For the processing capacities $P_i(t)$, we pushed the devices to 100\% load and measured the GPEF. 
We initialized the devices at 25\%-75\% loads, and treated the available remaining capacity as the receive buffer parameter $\theta_i(t)$. 

For the transmission costs, we measured GPEF spent on D2D offloading over WiFi. Our WiFi links, when only devoted to D2D offloading, consistently saturated at 12 Mbps. To simulate the effect of external tasks, we limit available bandwidth for D2D to 1, 6, and 9 Mbps. We then calculated the transmission resource budget for devices as transfer limits $\Psi_{i}(t)$, and modelled unit transfer costs $\psi_{i,j}(t)$ as normalized D2D latency.

\subsubsection{Datasets and large-scale network generation}
For FedL training, we use MNIST and Fashion-MNIST (F-MNIST)~\cite{fmnist} image classification datasets. We consider a CNN predictor composed of two convolutional layers with ReLU activation and dropout. The devices perform $\tau = 5$ rounds of gradient descent with a learning rate $\eta = 0.01$. 
Following \cite{tu2020network}, we generate network topologies with $N=100$ to $800$ devices using Erd{\"o}s–Rényi graphs with link formation probability (i.e., $A_{i,j} = 1$) of $0.1$. To produce local datasets across the nodes that are both overlapping and non-i.i.d, the datapoints at each node are chosen uniformly at random with replacement from datapoints among three labels (i.e., image classes). 
Differentiating the labels between devices captures dataset heterogeneity (i.e., from different devices collecting different labels). 
The number of initial datapoints $D_i(0)$ at each device follows a normal distribution with mean $\mu=(D_{\mathcal{N}}(0))N^{-1}$ and variance $\sigma^2 = 0.2\mu$. We further estimate the initial similarity weights $\lambda_{i,j}(0)$ based on the procedure discussed in Sec.~\ref{sss:graph}. 

\begin{table}[!t]
\caption{Global aggregations required by each scheme to reach a certain percentage of FedL model accuracy on MNIST (85\%) and F-MNIST (65\%) with $S = 6$ and varying $N$. Our smart sampling with offloading method consistently obtains the fastest training time.} 
\label{tab: mnist}
{\footnotesize
\begin{tabularx}{0.48\textwidth}{c *{6}{Y}}
\toprule[.2em]
\multirow{2}{*}{\bf{Sampling Method}} & \multicolumn{3}{c}{\bf{Devices (MNIST)}} & \multicolumn{3}{c}{\bf{Devices (F-MNIST)}} \\
\cmidrule(lr){2-4} \cmidrule{5-7}
& \bf{600} & \bf{700} & \bf{800} & \bf{600} & \bf{700} & \bf{800}\\
\midrule
Smart & 8 & 9 & 11 & 8 & 5 & 7\\
Random & \rule{5 mm}{0.2pt}& \rule{5 mm}{0.2pt} & \rule{5 mm}{0.2pt} & \rule{5 mm}{0.2pt} & 14 & 17\\
Heuristic & \rule{5 mm}{0.2pt} & 17 & \rule{5 mm}{0.2pt} & 14& 17& 7\\
\midrule
Smart w/ offload & 4 & 4 & 6 & 6 & 4 & 5\\
Random w/ offload & 11 & 14 & 10 & 14 & 12 & 13\\
Heuristic w/ offload & \rule{5 mm}{0.2pt} & 15 & \rule{5 mm}{0.2pt} & 9 & 11 & 5\\
\bottomrule
\end{tabularx}
}
\vspace{-5mm}
\end{table}


For the GCN-based sampling procedure, we train the model on small network realizations of ten devices. We consider sampling budgets of $S = 3$ to $6$, with thousands of training samples $E$ in each case. We save the resulting graph convolutional layer weights $\mathbf{Q}^{(1)}$ and $\mathbf{Q}^{(2)}$ for each choice of $S$ and reapply them on the larger target networks.


\begin{figure}[t]
\vspace{-2mm}
\centering
\hspace{-4mm}
\includegraphics[width=0.96\linewidth]{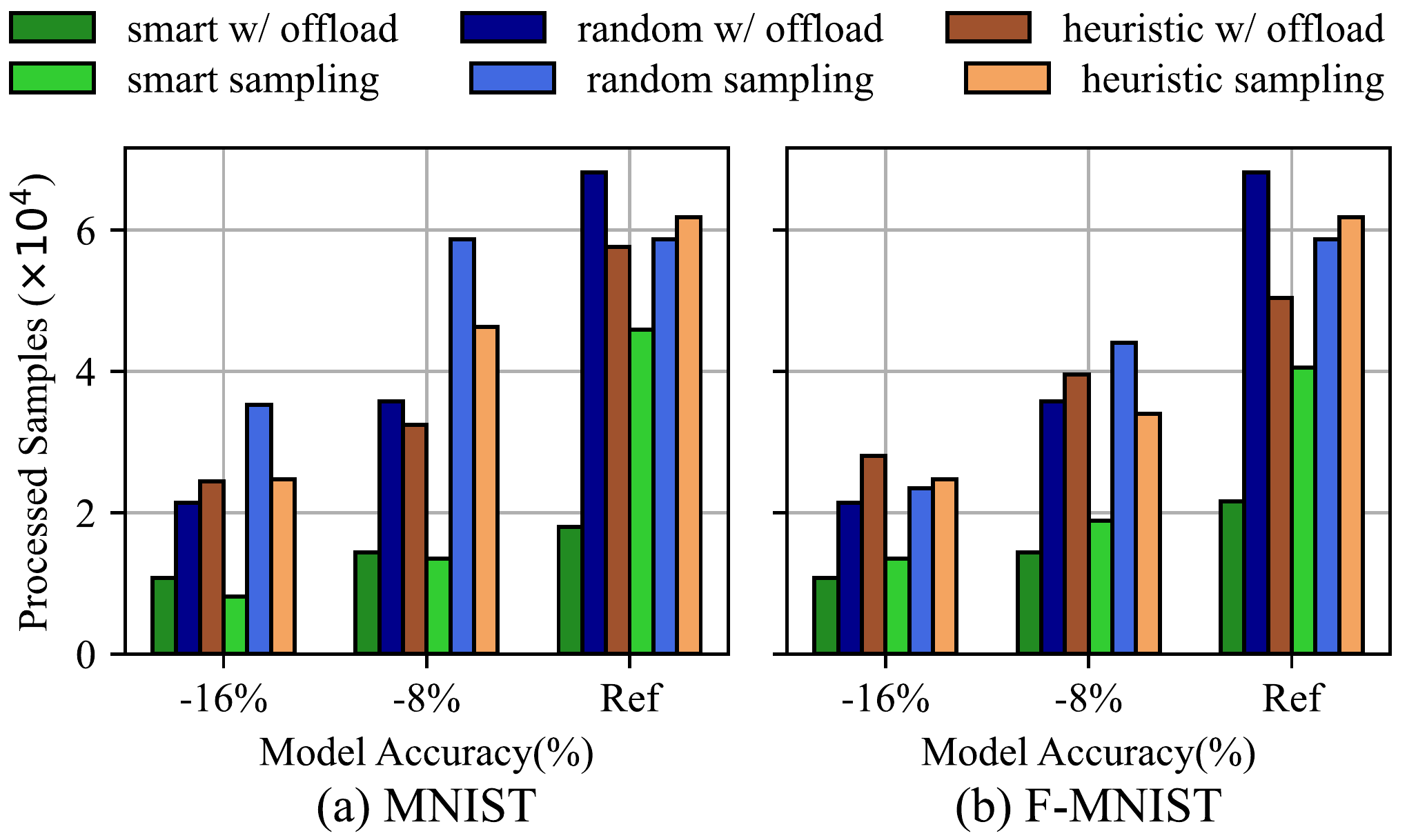}
\vspace{-2mm}
\caption{Number of samples processed by the schemes with and without offloading to reach within a certain percentage of a reference testing accuracy (Ref, 60\%). Our smart sampling with offloading methodology scales the best in terms of processing resources.}
\label{fig:bars}
\vspace{-6mm}
\end{figure}

\subsection{Results and Discussion} \label{ss:sim_baseline}

In the following experiments, we compare our methodology to several baseline sampling and offloading schemes. The three sampling strategies considered are \textit{smart}, \textit{random}, and \textit{heuristic}.
Smart sampling refers to our proposed method, random sampling is done by averaging the performance over five randomly sampled combinations of devices, and heuristic sampling selects the devices with the highest processing capacities.
Each of the three sampling schemes is either done \textit{without} or \textit{with offloading}. For smart sampling, our offloading methodology is used. For random sampling, we perform random offloading. For heuristic sampling, we perform a greedy offloading that maximizes the number of received data points for the device. A baseline of FedL with no sampling (i.e., all nodes active) and no offloading is also included.

\subsubsection{Model accuracy}
Figs.~\ref{fig:345_mnist} and~\ref{fig:345_fmnist} show FedL accuracy for both datasets in six different combinations ($S = 3, 4, 5$ and $N = 100, 200$). Overall, we see that the final accuracy obtained by our \textit{smart with (w/) offloading} scheme outperforms all of the other methods except for FedL in two cases (\ref{fig:345_mnist}a\&\ref{fig:345_mnist}d). The comparison with \textit{FedL all nodes} is remarkable as that leverages all of the devices in the network, while the sampling uses at most 5\% of them. 
Our smart data offloading methodology outperforms FedL due to two main reasons: (i) it minimizes data skew resulting from unbalanced label frequencies, and (ii) it ensures higher quality of local datasets at sampled nodes, which reduces bias caused by multiple local gradient descents.
Without offloading, the improvement obtained over the heuristic and random sampling strategies consistently exceeds 20\% for MNIST and 10\% for F-MNIST, which shows that sampling optimization still leads to considerable improvements when D2D is disabled.
On the other hand, our method with offloading obtains a substantial improvement over no offloading in most cases, whereas these differences are smaller for the heuristic and random methods. This emphasizes the importance of designing the sampling and offloading schemes for FedL jointly.


\subsubsection{Model convergence speed} We next compare the convergence speeds of our methodology to the other schemes in terms of the number of global aggregations needed to reach a certain percentage of the final accuracy of FedL with all nodes. Table~\ref{tab: mnist} compares the convergence speeds on MNIST (to reach 85\%) and F-MNIST (to reach 65\%), respectively, for $N = 600, 700, 800$ and $S = 6$. Overall, we see that our joint sampling and offloading methodology obtains significantly faster training speeds than the other methods, on average 40\% for MNIST and 50\% for F-MNIST. Enabling offloading is also seen to improve the convergence speeds of each sampling scheme; in fact, without offloading, several baseline cases fail to reach the given percentage of the FedL baseline.

\subsubsection{Resource utilization} Finally, we compare the resource utilization for the different schemes in terms of the total data processed across nodes in the network. Fig.~\ref{fig:bars} gives the results for each dataset, comparing the datapoints processed by methods with and without offloading to reach within a percentage of an arbitrary reference accuracy (60\%). We see that smart sampling (with or without offloading) outperforms the other schemes in all cases, which highlights the computational efficiency obtained by our method. As the accuracy level increases, our methodology constantly requires less datapoints compared to the other methods (on average 40\% fewer), emphasizing its ability to filter out duplicate data.


These experiments demonstrate that our joint optimization method exceeds baseline performances in terms of model accuracy, convergence speed, and resource utilization.

\section{Conclusion and Future Work}
 
\noindent In this paper, we developed a novel methodology to solve the joint sampling and D2D offloading optimization problem for FedL. Our theoretical analysis of the offloading subproblem produced new convergence bounds for FedL, and led to a sequential convex optimization solver. We then developed a GCN-based algorithm that determines the sampling strategy by learning the relationships between the network properties, the offloading topology, the sampling set, and the FedL accuracy. Our implementations using real-world datasets and IoT measurements from our testbed demonstrated that our methodology obtains significant improvements in terms of datapoints processed, training speed, and resulting model accuracy compared to several other algorithms, including FedL using all devices. Future investigations will consider the integration of realistic network characteristics on FedL.



\appendices
\vspace{-1mm}
\section{}\label{app:main1}
\noindent Since $\mathbf{v}_k(t)= \mathbf{v}_k(t-1) - \eta \nabla F(\mathbf{v}_k(t-1)|\mathcal{D}_{\mathcal{N}}(t))$, $\mathbf{w}_{\mathcal{S}}(t)  =  \mathbf{w}_{\mathcal{S}}(t-1) \hspace{-.5mm}- \hspace{-.5mm}\eta (\nabla F(\mathbf{w}_{\mathcal{S}}(t-1)|\mathcal{D}_{\mathcal{N}}(t)) \hspace{-.5mm}+\hspace{-.5mm} \zeta(\mathbf{w}_{\mathcal{S}}(t-1))$, we get:

\vspace{-5mm}
\small
\begin{align} \label{eq:t1_p1}
&\hspace{-2mm}\big\Vert \mathbf{w}_{\mathcal{S}}(t) - \mathbf{v}_{k}(t) \big\Vert \hspace{-1mm} = \hspace{-1mm} \Big\Vert \mathbf{w}_{\mathcal{S}}(t \hspace{-0.5mm}- \hspace{-0.5mm} 1) - \mathbf{v}_k(t\hspace{-0.5mm}-\hspace{-0.5mm}1) - \eta \zeta(\mathbf{w}_{\mathcal{S}}(t\hspace{-0.5mm}-\hspace{-0.5mm}1)) \nonumber\\
&\hspace{-2mm}- \eta\nabla F(\mathbf{w}_{\mathcal{S}}(t-1)|\mathcal{D}_{\mathcal{N}}(t)) +\eta \nabla F(\mathbf{v}_k(t-1) \vert \mathcal{D}_{\mathcal{N}}(t)) \Big\Vert.\hspace{-2mm}
\end{align}
\normalsize

\noindent We simplify \eqref{eq:t1_p1} through the following steps:
\vspace{-5mm}

\small
\begin{align}\label{eq:t1_bigboy}
& \Vert \mathbf{w}_{\mathcal{S}}(t) - \mathbf{v}_{k}(t) \Vert \overset{(a)}{\leq} \Vert \mathbf{w}_{\mathcal{S}}(t-1)-\mathbf{v}_k(t-1) \Vert \nonumber \\
&  + \eta \sum_{i \in \mathcal{N}} \frac{D_i(t-1)}{D_{\mathcal{N}}(t)} \Vert\nabla F(\mathbf{w}_{\mathcal{S}}(t-1)|\mathcal{D}_i(t-1)) \nonumber \\
& - \nabla F(\mathbf{v}_k(t-1)|\mathcal{D}_i(t-1)) \Vert +\eta \Vert \zeta(\mathbf{w}_{\mathcal{S}}(t-1)) \Vert \nonumber \\
&\overset{(b)}{\leq} \Vert \mathbf{w}_{\mathcal{S}}(t-1)-\mathbf{v}_k(t-1) \Vert + \eta \Vert \zeta(\mathbf{w}_{\mathcal{S}}(t-1))\Vert  \nonumber \\ 
&+\eta\beta \sum_{j \in \mathcal{N}}\frac{D_j(t-1)}{D_{\mathcal{N}}(t) D_{\mathcal{S}}(t) } \sum_{i \in \mathcal{S}}D_i(t-1)\Vert \mathbf{w}_{\mathcal{S}}(t-1) - \mathbf{v}_k(t-1)\Vert \nonumber \\
&\overset{(c)}{\leq} \Vert \mathbf{w}_{\mathcal{S}}(t-1) - \mathbf{v}_k(t-1) \Vert + \eta \Vert \zeta(\mathbf{w}_{\mathcal{S}}(t-1))\Vert \nonumber \\ 
&+\sum_{j \in \mathcal{N}}\frac{D_j(t-1)}{D_{\mathcal{N}}(t) D_{\mathcal{S}}(t) } \sum_{i \in \mathcal{S}}D_i(t-1)\frac{\delta_i(t)}{\beta}(2^{t-1-(k-1)\tau}-1)  \nonumber \\
&\overset{(d)}{\leq} \Vert \mathbf{w}_{\mathcal{S}}(t-1)-\mathbf{v}_k(t-1) \Vert + \frac{1}{\beta} \Vert \zeta(\mathbf{w}_{\mathcal{S}}(t-1))\Vert  \nonumber \\ 
&+\sum_{j \in \mathcal{N}}\frac{D_j(t-1)}{D_{\mathcal{N}}(t)} \frac{\delta_{\mathcal{S}}(t)}{\beta}(2^{t-1-(k-1)\tau}-1), 
\end{align}
\normalsize
\vspace{-4mm}

\noindent where $(a)$ results from expanding $\nabla F(\mathbf{v}_k(t-1)|\mathcal{D}_{\mathcal{N}}(t))$ and applying the triangle inequality repeatedly, $(b)$ follows from using the $\beta$-smoothness of the loss function and the triangle inequality, $(c)$ applies Lemma 3 from \cite{wang2019adaptive}, and $(d)$ uses the expanded form of $\delta_{\mathcal{S}}(t)$ in \eqref{th1:3}. We then rearrange~\eqref{eq:t1_bigboy}:
\begin{equation}
\begin{aligned}
&\Vert \mathbf{w}_{\mathcal{S}}(t)-\mathbf{v}_k(t) \Vert - \Vert \mathbf{w}_{\mathcal{S}}(t-1)-\mathbf{v}_k(t-1) \Vert  \\ 
&\leq \frac{\Upsilon(t,k)}{\beta} +\frac{1}{\beta} \Vert \zeta(\mathbf{w}_{\mathcal{S}}(t-1)) \Vert.
\end{aligned}
\end{equation}
\noindent 
Since $\Vert \mathbf{w}_{\mathcal{S}}(t) - \mathbf{v}_{k}(t)\Vert = 0$ when re-synchronization occurs at $t = k\tau$ , ~$\forall k \in \{1,\cdots,K \}$, we express $\Vert \mathbf{w}_{\mathcal{S}}(t)-\mathbf{v}_k(t) \Vert$ as:
\small
\begin{align}
&\Vert \mathbf{w}_{\mathcal{S}}(t)- \mathbf{v}_k(t) \Vert  = \sum_{y = (k-1)\tau+1}^{t}  \Vert \mathbf{w}_{\mathcal{S}}(y) - \mathbf{v}_k(y) \Vert -  \Vert \mathbf{w}_{\mathcal{S}}(y-1)\nonumber \\ &-  \mathbf{v}_k(y-1) \Vert 
 \leq \frac{1}{\beta} \hspace{-2mm} \sum_{y = (k-1)\tau+1}^{t} \left(\Upsilon(y,k) +\Vert \zeta(\mathbf{w}_{\mathcal{S}}(y-1)) \Vert \right). \label{eq:sample}\hspace{-2mm}
\end{align}
\normalsize

\vspace{-1mm}
\section{}\label{app:c1}
\noindent We first define $\theta_k(t) \triangleq F(\mathbf{v}_k(t))-F(\mathbf{w}^*(t)) \geq \epsilon$.
Since $F$ is $L$-Lipschitz, we apply the result of Theorem 1, and Lemmas 2 and 6 from \cite{wang2019adaptive} to obtain:
\vspace{-1.2mm}
\small
\begin{align}
& \theta_{\hat{K}+1}(t)^{-1} - \theta_{1}(0)^{-1} = \left(\theta_{\hat{K}+1}(t)^{-1} - \theta_{k+1}(\hat{K}\tau)^{-1} \right) \nonumber \\
& + \left(\theta_{\hat{K}+1}(\hat{K}\tau)^{-1} - \theta_{\hat{K}}(k\tau)^{-1} \right) + \left(\theta_{\hat{K}}(\hat{K}\tau)^{-1} - \theta_{1}(0)^{-1} \right)\nonumber\\
& \geq \left((t-\hat{K}\tau)\xi \eta \left(1-\frac{\beta \eta}{2}\right)\right) + \hat{K}\tau \xi \eta \left(1- \frac{\beta \eta}{2}\right) - \frac{L}{\beta \epsilon^2} \hat{\Upsilon}(\hat{K}) \nonumber\\
& - (\hat{K}-1) \frac{L}{\beta \epsilon^2}\hat{\Upsilon}(\hat{K}) 
= t\xi \eta \left(1-\frac{\beta \eta}{2}\right) - \hat{K} \frac{L}{\beta \epsilon^2} \hat{\Upsilon}(\hat{K}), \label{eq:c1eq1}\hspace{-3mm}
\end{align}
\normalsize
\noindent where $\hat{\Upsilon}(\hat{K}) \triangleq \sum_{y=(\hat{K}-1)\tau+1}^{\hat{K}\tau} \Upsilon(\hat{K},y) + \Vert \zeta(\mathbf{w}_{\mathcal{S}}(y-1))\Vert $. Using the assumptions from Lemma~2 of~\cite{wang2019adaptive} and the fact that $t \geq \hat{K}\tau$, the RHS of~\eqref{eq:c1eq1} is strictly positive, implying that $\hat{\Upsilon} (\hat{K}) \ll 1$ as $t\xi \eta (1-\frac{\beta \eta}{2})$ is very small. Next, defining $\varrho(t) \triangleq F(\mathbf{w}_{\mathcal{S}}(t)|\mathcal{D}_{\mathcal{N}}(t)) - F(\mathbf{w}^*(t)|\mathcal{D}_{\mathcal{N}}(t)) \geq \epsilon$, we get:

\small
\begin{equation} \label{eq:cleq2}
\hspace{-0mm}
\begin{aligned}
&\varrho(t)^{-1} - \frac{1}{\theta_{\hat{K}+1}(t)}  = \frac{F(\mathbf{v}_{\hat{K}}(t)|\mathcal{D}_{\mathcal{N}}(t)) - F(\mathbf{w}_{\mathcal{S}}(t)|\mathcal{D}_{\mathcal{N}}(t))}{\theta_{\hat{K}+1}(t) \varrho(t) } \\
& \geq \frac{-L}{\beta \epsilon^2} \sum_{y=(\hat{K}-1)\tau+1}^{\hat{K}\tau} \Upsilon(\hat{K},y) + \Vert \zeta(\mathbf{w}_{\mathcal{S}}(y-1)) \Vert.
\end{aligned}
\hspace{-6mm}
\end{equation}
\normalsize

\noindent Combining \eqref{eq:c1eq1} and \eqref{eq:cleq2}, and since $\theta_1(0) > 0$, we get $\varrho(t)^{-1} \geq t \xi \eta \left(1 - \frac{\beta \eta}{2}\right) - \frac{(\hat{K}+1)L}{\beta \epsilon^2} \hat{\Upsilon}(\hat{K})$, taking the reciprocal of which leads to \eqref{eq:cl1_result}.



\vspace{-1mm}
\section{}\label{app:prop1}

\noindent Since $\nabla F(\mathbf{w}_{\mathcal{S}}(t)|\mathcal{D}_i(t))$ is the average of $\nabla F(\mathbf{w}_{\mathcal{S}}(t),x_d,y_d)$, $\forall (x_d,y_d) \in \mathcal{D}_i(t)$, we apply the central limit theorem to view $\nabla F(\mathbf{w}_{\mathcal{S}}(t)|\mathcal{D}_i(t))$ as $D_i(t)$ samples of $\nabla F(\mathbf{w}_{\mathcal{S}}(t),x_d,y_d)$ from a distribution with mean $\nabla F(\mathbf{w}_{\mathcal{S}}(t)|\mathcal{D}_{\mathcal{N}}(t))$. Then, \eqref{eq:lemma} can be upper bounded using the definition in~\eqref{def:e}:
\vspace{-3mm}

\small
\begin{align}
&\Vert\nabla F(\mathbf{w}_{\mathcal{S}}(t)|\mathcal{D}_i(t)) - \nabla F(\mathbf{w}_{\mathcal{S}}(t)|\mathcal{D}_{\mathcal{N}}(t)) - \zeta(\mathbf{w}_{\mathcal{S}}(t))\Vert \nonumber \\
&\leq \Vert \zeta(\mathbf{w}_{\mathcal{S}}(t)) \Vert  + \frac{\gamma}{\sqrt{D_i(t)}} \leq \Bigg\Vert \frac{-1}{D_{\mathcal{N}}(t)} \sum_{i \in \hat{\mathcal{S}}} D_i(t) \nabla F(\mathbf{w}_{\mathcal{S}}(t)|\mathcal{D}_i(t))\nonumber  \\
& + \frac{D_{\mathcal{N}}(t)-D_{\mathcal{S}}(t)}{D_{\mathcal{N}}(t) D_{\mathcal{S}}(t)} \sum_{i \in \mathcal{S}} D_i(t) \nabla F(\mathbf{w}_{\mathcal{S}}(t)|\mathcal{D}_i(t))   \Bigg\Vert + \frac{\gamma}{\sqrt{D_i(t)}}.
\hspace{-4mm}
\end{align}
\normalsize
Applying the triangle inequality on the above gives the result. 

\newpage

\balance
\bibliographystyle{IEEEtran}
\bibliography{References}

\end{document}